\documentclass[a4paper]{article}
\usepackage[utf8]{inputenc}
\usepackage[english]{babel}

\usepackage{csquotes}
\usepackage[backend=biber]{biblatex}
\usepackage{amsmath}
\usepackage{amssymb}
\usepackage{amsthm}
\usepackage{bm}

\usepackage{xcolor}
\usepackage{enumerate}
\usepackage{listings}

\newtheorem{theorem}{Theorem}[section]
\newtheorem{lemma}[theorem]{Lemma}
\newtheorem{corollary}[theorem]{Corollary}
\newtheorem{proposition}[theorem]{Proposition}
\newtheorem{definition}[theorem]{Definition}

\newtheorem{remark}[theorem]{Remark}

\title{Characterizing NC${}^1$ with Typed Monoids} 

\author{
    Anuj Dawar\thanks{Research funded in part by UK Research and Innovation (UKRI) under the UK government’s Horizon Europe funding guarantee: grant number EP/X028259/1.}\\
    \texttt{anuj.dawar@cl.cam.ac.uk}
    \and
    Aidan T. Evans\footnote{Corresponding Author}\\
    \texttt{ate26@cam.ac.uk}
}
\date{
    Department of Computer Science and Technology\\ 
    University of Cambridge, United Kingdom
}

\begin{document}

\newcommand{\class}[1]{\mathrm{#1}}
\newcommand{\PT}{\class{P}}
\newcommand{\NP}{\class{NP}}
\newcommand{\TC}{\class{TC}}
\newcommand{\NC}{\class{NC}}

\newcommand{\logic}[1]{\mathrm{#1}}
\newcommand{\FO}{\logic{FO}}
\newcommand{\str}[1]{\mathfrak{#1}}
\newcommand{\tup}[1]{\overline{#1}}
\newcommand{\fin}{\mathrm{fin}}
\newcommand{\syn}{\mathrm{syn}}

\maketitle

\begin{abstract}
    Krebs et al.~(2007) gave a characterization of the complexity class $\TC^0$ as the class of languages recognized by a certain class of typed monoids.  The notion of typed monoid was introduced to extend methods of algebraic automata theory to infinite monoids and hence characterize classes beyond the regular languages.  We advance this line of work beyond $\TC^0$ by giving a characterization of $\NC^1$.  This is obtained by first showing that $\NC^1$ can be defined as the languages expressible in an extension of first-order logic using only unary quantifiers over regular languages.  The expressibility result is a consequence of a general result showing that finite monoid multiplication quantifiers of higher dimension can be replaced with unary quantifiers in the context of interpretations over strings, which also answers a question of Lautemann et al.~(2001).  We establish this collapse result for a much more general class of interpretations  using results on interpretations due to  Boja\'{n}czyk et al. (2019), which may be of independent interest.
\end{abstract}

\section{Introduction}

Much work in theoretical computer science is concerned with studying classes of formal languages, whether these are classes defined in terms of grammars and expressions, such as the class of regular or context-free languages, or whether they are \emph{complexity classes} such as $\PT$ and $\NP$, defined by resource bounds on machine models.  Indeed, the distinction between these is largely historical as most classes of interest admit different characterizations based on machine models, grammars, logical definability, or algebraic expressions. The class of regular languages can be characterized as the languages accepted by linear-time-bounded single-tape Turing machines~\cite{hennie1965one} while $\PT$ can be characterized without reference to resources as the languages recognized by multi-head two-way pushdown automata~\cite{cook1971characterizations}.  The advantage of the variety of characterizations is, of course, the fact that these bring with them different mathematical toolkits that can be brought to the study of the classes.

The class of regular languages has arguably the richest theory in this sense of diversity of characterizations. Most students of computer science learn of the equivalence of deterministic and nondeterministic finite automata, regular languages and linear grammars and many also know that the regular languages are exactly those definable in monadic second-order logic with an order predicate.  Perhaps the most productive approach to the study of regular languages is via their connection to finite monoids.  Every language $L$ has a syntactic monoid, which is finite if, and only if, $L$ is regular.  Moreover, closure properties of classes of regular languages relate to natural closure properties of classes of monoids, via Eilenberg's Correspondence Theorem~\cite{eilenberg1976automata}.  Finally, we have \emph{Krohn-Rhodes theory} which decomposes regular languages into elementary components; such a decomposition has found applications not only in semigroup theory but also, for example, in physics \cite[Sections 4.8a, 4.10]{rhodes2010applications} and the modeling of biochemical reactions \cite[Chapter 6, Part I]{rhodes2010applications}. All together, these tools give rise to \emph{algebraic automata theory}---which leads to the definition of natural subclasses of the class of regular languages, to effective decision procedures for automata recognizing such classes, and to separation results.

When it comes to studying computational complexity, we are mainly interested in classes of languages richer than just the regular languages.  Thus, the syntactic monoids of the languages are not necessarily finite and the extensive tools of Krohn-Rhodes theory are not available to study them.  Nonetheless, some attempts have been made to extend the methods of algebraic automata theory to classes beyond the regular languages.  Most significant is the work of Krebs and collaborators~\cite{behle2007linear,behle2011typed,krebs2007characterizing,krebs2008typed,cano2021positive}, which introduces the notion of \emph{typed monoids}.  The idea is to allow for languages with infinite syntactic monoids, but limit the languages they recognize by associating with the monoids a finite collection of types.  This allows for the formulation of a version of Eilenberg's Correspondence theorem associating closure properties on classes of typed monoids with corresponding closure properties of classes of languages.  In particular, this implies that most complexity classes of interest can be uniquely characterized in terms of an associated class of typed monoids~\cite{behle2011typed}.  An explicit description of the class characterizing uniform $\TC^0$ is given by Krebs et al.~\cite{krebs2007characterizing,krebs2008typed} and, furthermore, provides a decomposition of $\TC^0$ into elementary components analogous to the Krohn-Rhodes theory for regular languages.  This is obtained through a general method which allows us to construct typed monoids corresponding to \emph{unary quantifiers} defined from specific languages~\cite{krebs2008typed} (see also Theorem~\ref{thm:logcircalgequiv} below).

In this paper, we extend this work to obtain a characterization and decomposition of (uniform) $\NC^1$ as the class of languages recognized by the collection of typed monoids obtained as the closure under \emph{ordered strong block products} of three typed monoids: the group of integers with types for positive and non-positive integers; the monoid of the natural numbers with types for the square numbers and non-square numbers; and a finite non-solvable group such as $S_5$ with a type for each subset of the group.  Full definitions of these terms follow below.  Our result is obtained by first characterizing  $\NC^1$ in terms of logical definability in an extension of first-order logic with only unary quantifiers.  It is known that any regular language whose syntactic monoid is a non-solvable group is complete for $\NC^1$ under reductions definable in first-order logic with arithmetic predicates ($\FO[+,\times]$)~\cite{barrington1990uniformity}.  From this, we know we can describe $\NC^1$ as the class of languages definable in an extension of $\FO[+,\times]$ with quantifiers (of arbitrary arity) associated with the regular languages recognized by the monoid $S_5$.  We show that the family of such quantifiers associated with any finite monoid can be replaced with just the unary quantifiers.  This answers a question left open in Lautemann et al.~\cite{lautemann2001descriptive} and allows us to obtain the sought after algebraic characterization.

For the purpose of establishing the characterization of $\NC^1$, it suffices to consider quantifiers applied to interpretations in which tuples of elements are ordered lexicographically.  However, we show that the arity collapse of quantifiers associated with finite monoids holds more generally, for any first-order definable linear orders on tuples.  To show this, we levarage a characterization of first-order interpretations due to Boja\'{n}czyk  et al.~\cite{bojanczyk2019string,bojanczyk2022transducers}.

We begin in Section~\ref{sec:prelim}, by covering relevant background on semigroup theory, typed monoids, and multiplication quantifiers (some of which is relegated to the appendix in the interests of space).  In Section~\ref{sec:mult}, we establish the technical result showing that quantifiers of higher arity over a regular language $L$ can be defined using just unary quantifiers over the syntactic monoid of $L$, when the quantifiers are applied to interpretations with a lexicographic order on tuples.  In  Section~\ref{sec:alg}, we apply this to obtain the algebraic characterization of  $\NC^1$.  finally, in Section~\ref{sec:nonlex}, we establish the arity collapse for quantifiers in the context of more general interpretations.

\section{Preliminaries}\label{sec:prelim}

We assume the reader is familiar with basic concepts of formal language theory, automata theory, complexity theory, and logic.  We quickly review definitions we need in order to fix notation and establish conventions.

We write $\mathbb{Z}$ for the set of integers, $\mathbb{N}$ for the set of natural numbers (including $0$), and $\mathbb{Z}^+$ for the set of positive integers.
We write $[n]$ for the set of integers $\{1,\ldots,n\}$ and $\mathbb{S}$ for the set of \emph{square} integers.  That is, $\mathbb{S} = \{ x \in \mathbb{Z}^+ \mid x = y^2 \text{ for some } y \in \mathbb{Z}\}$.

For a fixed $n \in \mathbb{Z}^+$ and an integer $i \in [n]$, we define the \emph{$n$-bit one-hot encoding} of $i$ to be the binary string $b \in \{0,1\}^n$ such that $b_j = 1$ if, and only if,  $j = i$.  For any set $S$, we write $\wp(S)$ for the powerset of $S$.

The complexity classes $\TC^0$ and $\NC^1$ are classes of languages defined in terms of circuits.  We are only interested in the \emph{uniform} versions of these classes.  Specifically, $\TC^0$ is the class of languages recognized by a uniform family of circuits of polynomial size and constant depth using And, Or, Not and Majority gates of arbitrary fan-in; and $\NC^1$ is the class of languages recognized by a uniform family of circuits of polynomial size and logarithmic depth using And, Or and Not  gates of fan-in at most 2.  Since gates (including majority gates) of unbounded fan-in can be simulated by circuits of logarithmic depth of fan-in 2 using just the standard Boolean basis, we have $\TC^0 \subseteq \NC^1$ and the separation of the two classes is open.  The requirement of \emph{uniformity} here means that the circuits of input size $n$ are easily computed from $n$.  It is standard to take ``easily computed'' to mean $\textsc{DLogTime}$-uniform, though the classes are robust under varying the definition (see~\cite{vollmer1999introduction}).

\subsection{Semigroups, Monoids, and Groups}\label{sec:semigroups}

A \emph{semigroup} $(S, \cdot)$ is a set $S$ equipped with an \emph{associative} binary operation.  We call a semigroup \emph{finite} if $S$ is finite. 
Context permitting, we may refer to a semigroup $(S, \cdot)$ simply by its underlying set $S$.
  
A \emph{monoid} $(M, \cdot)$ is a semigroup with a distinguished element $1_M \in M$ such that for all $m \in M$, $1_M \cdot m = m \cdot 1_M = m$.  We call $1_M$ the \emph{identity} or \emph{neutral} element of $M$.  A \emph{group} $(G, \cdot)$ is a monoid such that for every $g \in G$, there exists an element $g^{-1} \in G$ such that $g \cdot g^{-1} = g^{-1} \cdot g = 1$.  We call $g^{-1}$ the \emph{inverse} of $g$.  

Note that $(\mathbb{Z},+)$ is a group, $(\mathbb{N},+)$ is a monoid but not a group and $(\mathbb{Z}^+,+)$ is a semigroup but not a monoid.  In the first two cases, the identity element is $0$.  When we refer to the monoids $\mathbb{Z}$ or $\mathbb{N}$ we assume that the operation referred to is standard addition.%

For a semigroup $(S, \cdot)$, we say that a set $G \subseteq S$ \emph{generates} $S$ if $S$ is equal to the closure of $G$ under $\cdot$; we denote this by $S = \langle G \rangle_{\cdot}$, or, simply, $\langle G \rangle$ if the operation is clear from context, and call $G$ a \emph{generating set of} $S$. 
We say that $S$ is \emph{finitely generated} if there exists a finite generating set of $S$.  All semigroups we consider are finitely generated.  Note that $\mathbb{Z}^+$ is generated by $\{1\}$, $\mathbb{N}$ by $\{0,1\}$ and $\mathbb{Z}$ by $\{1,-1\}$.

We write $U_1$ for the monoid $(\{0,1\},\cdot)$ where the binary operation is the standard multiplication.  Note that $1$ is the identity element here.
For any set $S$, we denote by $S^+$ the set of non-empty finite strings over $S$ and by $S^*$ the set of all finite strings over $S$.  Equipped with the concatenation operation on strings, which we denote by either $\circ$ or simply juxtaposition, $S^*$ is a monoid and $S^+$ is a semigroup but not a monoid.  We refer to these as the \emph{free monoid} and \emph{free semigroup} over $S$, respectively.  Note that $S$ is a set of generators for $S^+$ and $S \cup \{\epsilon\}$ is a set of generators for $S^*$.

A monoid homomorphism from a monoid $(S, \cdot_S)$ to a monoid $(T, \cdot_T)$ is a function $h : S \rightarrow T$ such that for all $s_1, s_2 \in S$, $h(s_1 \cdot_S s_2) = h(s_1) \cdot_T h(s_2)$ and $h(1_S) = 1_T$.     A \emph{congruence} on a monoid $(M, \cdot)$ is an equivalence relation $\sim$ on $M$ such that for all $a,b,c,d \in M$, if $a \sim b$ and $c \sim d$, then $a \cdot c \sim b \cdot d$.  We denote by $M/{\sim}$ the set of equivalence classes of $\sim$ on $M$.  We denote by $[a]_{\sim}$, or simply $[a]$, the equivalence class of $a \in M$ under $\sim$.  Any congruence $\sim$ gives rise to the \emph{quotient monoid} of $M$ by $\sim$, namely the monoid $(M/{\sim}, \star)$ where for $[a],[b] \in M/{\sim}$, $[a] \star [b] = [a \cdot b]$.  The map $\eta : M \rightarrow M/{\sim}$ defined by $\eta(a) = [a]$ is then a homomorphism, known as the  \emph{canonical homomorphism} of $M$ onto $M/{\sim}$.

\subsection{Language Recognition with Monoids}
\label{sec:backgroundsemi}

Here we recall the definition of syntactic congruences and syntactic monoids.

\begin{definition}\label{def:syncong}
    For a language $L \subseteq \Sigma^*$ over an alphabet $\Sigma$, we define the \emph{syntactic congruence} of $L$ as the equivalence relation $\sim_L$ on $\Sigma^*$ such that for all $x, y \in \Sigma^*$, $x \sim_L y$ if, and only if, for all $w,v \in \Sigma^*$, $wxv \in L$ if, and only if, $wyv \in L$.
\end{definition}

It is easily seen that this relation is a congruence on the free monoid $\Sigma^*$.  The quotient monoid $\Sigma^*/{\sim_L}$ is known as the \emph{syntactic monoid} of $L$.  More generally, we say that a monoid $M$ \emph{recognizes} a language $L$ if there is a homomorphism $h: \Sigma^* \rightarrow M$ and a set $A \subseteq M$ such that $L = h^{-1}(A)$.  It is easily seen that the syntactic monoid of $L$ recognizes $L$.  A language is regular if, and only if, its syntactic monoid is finite~\cite{straubing1994finite}.

\subsection{Logics and Quantifiers}\label{sec:backgroundmult}

We assume familiarity with the basic syntax and semantics of first-order logic.  In this paper, the logic is always interpreted in finite relational structures.  We generally denote structures by Fraktur letters, $\str{A}$, $\str{B}$, etc., and the corresponding universe of the structure is denoted $|\str{A}|$, $|\str{B}|$, etc.
We are almost exclusively interested in \emph{strings} over a finite alphabet.  Thus, fix an alphabet $\Sigma$.  A $\Sigma$-string is then a structure $\str{A}$ whose universe $|\str{A}|$ is linearly ordered by a binary relation $<$ and which interprets a set of unary relation symbols $(R_{\sigma})_{\sigma \in \Sigma}$ such that for each element $a \in |\str{A}|$ there is a unique $\sigma \in \Sigma$ such that $a$ is in the interpretation of $R_{\sigma}$.

More generally, let $\tau$ be any relational vocabulary consisting of a binary relation symbol $<$ and unary relation symbols $R_1,\ldots,R_k$.  We can associate with any $\tau$-structure in which $<$ is a linear order a string over an alphabet of size $2^k$ as formalized in the following definition.
\begin{definition}\label{def:associated}
  For $\tau$ a relational vocabulary consisting of a binary relation symbol $<$ and unary relation symbols $R_1,\ldots,R_k$, and $\str{A}$ a $\tau$-structure with $n$ elements that interprets the symbol $<$ as a linear order of its universe, we define the string $w_{\str{A}}$ \emph{associated with} $\str{A}$ as the string of length $n$ over the alphabet $\Sigma_k = \{0,1\}^k$ of size $2^k$ so that if $a$ is the $i$th element  of $w_{\str{A}}$, then $a$ is the $k$-tuple where $a_j =1$ if, and only if, $R_j$ holds at the $i$th element of $\str{A}$.
\end{definition}

\noindent For example, say $\tau = \{R_1, R_2, R_3\}$ and $\mathfrak{A} = ([4], <^\mathfrak{A}, R_1^{\mathfrak{A}}, R_2^{\mathfrak{A}}, R_3^{\mathfrak{A}})$ where $1 < 2 < 3 < 4$, $R_1^{\mathfrak{A}} = \{1,3\}$, $R_2^{\mathfrak{A}} = \{2,4\}$, and $R_3^{\mathfrak{A}} = \{1,2,3\}$.  Then, the string $w_{\str{A}} = (101)(011)(101)(010)$. 

Using the above, we can associate a language with any isomorphism-closed class of structures over the vocabulary $\tau$.  We formalize this definition for future use.
\begin{definition}\label{def:associated-lang}
  For $\tau$ a relational vocabulary consisting of a binary relation symbol $<$ and unary relation symbols $R_1,\ldots,R_k$, and $\mathcal{A}$ a class of $\tau$-structures where every structure interprets $<$ as a linear order, we define the language $L_{\mathcal{A}}$ over the alphabet $\Sigma_k = \{0,1\}^k$ to be
  $$L_{\mathcal{A}} = \{ w_{\str{A}} \mid \str{A} \in \mathcal{A} \}.$$

  Conversely, for any language $L$ over the alphabet $\Sigma_k$, we define the class of $\tau$-structures $\mathcal{S}_L$ to be
  $$\mathcal{S}_L = \{ \str{A} \mid w_{\str{A}} \in L \}.$$
\end{definition}
It is worth noting that the above operations are not inverses, in the sense that it is not the case that $S_{L_{\mathcal{A}}} = \mathcal{A}$.  This is because while $\mathcal{A}$ is a class of structures using $k$ unary relation symbols, $S_{L_{\mathcal{A}}}$ is in a vocabulary using $2^k$ unary symbols.

As the elements of a string $\str{A}$ are linearly ordered, we can identify them with an initial segment $\{1,\ldots,n\}$ of the positive integers.  In other words, we treat a string with universe $\{1,\ldots,n\}$ and the standard order on these elements as a canonical representative of its isomorphism class.  In addition to the order predicate, we may allow other \emph{numerical predicates} to appear in formulas of our logics.  These are predicates whose meaning is completely determined by the size $n$ of the structure and the ordering of its elements.  In particular, we have ternary predicates $+$ and $\times$ for the partial addition and multiplication functions.

An insight due to Lindstr\"om allows us to define a \emph{quantifier} from any isomorphism-closed class of structures (see~\cite{Ebb85}).  We adopt the terminology of Ebbginhaus and Flum~\cite[Chapter~12]{EF99} and provide a brief overview now.
Consider a relational vocabulary $\tau = \{ R_1,\ldots,R_l\}$, where for each $i$, $R_i$ is a relation symbol of arity $r_i$. For any vocabulary $\sigma$ and positive integer $d$, an \emph{interpretation} of $\tau$ in $\sigma$ of dimension $d$ is a tuple of formulas $I = (\phi_1(\tup{x}_1),\dots, \phi_l(\tup{x}_l))$ of vocabulary $\sigma$ where $\phi_i$ is associated with a tuple $\tup{x}_i$ of variables of length $dr_i$.  Suppose we are given a $\sigma$-structure $\str{A}$ and an assignment $\alpha$ that takes variables to elements of $\str{A}$.  Then let $\phi_i^{\str{A},\alpha}$ denote the relation of arity $dr_i$ consisting of the set of tuples $\{\tup{a} \in |\str{A}|^{dr_i} \mid \str{A} \models \phi_i[\alpha[\tup{x}_i/\tup{a}]] \}$.  Then, the interpretation $I$ defines a map that takes a $\sigma$-structure $\str{A}$, along with an assignment $\alpha$ to the  $\tau$-structure $I(\str{A},\alpha)$ with universe $|\str{A}|^d$ where the interpretation of $R_i$ is the set $\phi_i^{\str{A},\alpha}$, seen as a relation of arity $r_i$ on $|\str{A}|^d$.

Given a class of $\tau$-structures $Q$ and any positive integer $d$, we have a quantifier $Q_d$.  In a logic with $Q_d$, we can form formulas of the form
$$Q_d \tup{x}_1\cdots \tup{x}_l (\phi_1,\ldots,\phi_l)$$
whenever $I = (\phi_1(\tup{x}_1),\dots, \phi_l(\tup{x}_l))$  is an interpretation of dimension $d$.  In this formula, occurrences in the subformula $\phi_i$  of variables among $x_i$ are bound.   The semantics of this quantifier are given by the rule that 
$Q_d \tup{x}_1\cdots \tup{x}_l (\phi_1,\ldots,\phi_l)$ is true in a structure $\str{A}$ under some assignment $\alpha$ of values to the free variables if the $\tau$-structure $I(\str{A},\alpha)$ is in $Q$.
We can understand $Q_d$ as the $d$th vectorization of the quantifier $Q$ (see~\cite[Def.~12.1.6]{EF99}). When $d=1$, we may omit the subscript.  A quantifier of the form $Q_1$ is called \emph{unary}.

The standard first-order quantifiers: $\exists$ and $\forall$ can be seen as special cases of Lindstr\"om quantifiers in a vocabulary with one unary relation $U$.  The existential quantifier consists of all structures $(A,U)$ where $U \subseteq A$ is non-empty and the universal quantifier  consists of all structures $(A,U)$ where $U = A$.


We are particularly interested in interpretations $I$ where both $\sigma$ and $\tau$ are vocabularies of strings.  These are also known in the literature as \emph{string-to-string transducers}. (See \cite{bojanczyk2022transducers} for an example of how transducers may have many representations.) Therefore, $I$ must define an interpretation of not only each $R_i \in \tau$ (using the formula $\phi_i$) but also a $\sigma$-formula $\phi_<$ defining the linear order on the universe of the $\tau$-structure $I(\str{A},\alpha)$.  We are particularly interested in the case where the order defined is  the lexicographic order on $d$-tuples of $|\str{A}|$ induced by the order in $\str{A}$.  This order is easily defined by a (quantifier-free) first-order formula, and we often omit it from the description of $I$.  

When we consider interpretations that define a linear order other than the lexicographic order, we explicitly include the formula defining the order and thus the formula is $Q_d \tup{x}\tup{y} (\phi_<(\tup{x}), \phi_1(\tup{y}),\ldots,\phi_l(\tup{y}))$, where the interpretation is of dimension $d$ and so $|\tup{x}| = 2d$ and $|\tup{y}| = d$.

Finally, we now introduce some notation we use in the rest of the paper for various logics formed by combining particular choices of quantifiers and numerical predicates.
\begin{definition}
  For a set of quantifiers $\mathfrak{Q}$ and numerical predicates $\mathfrak{N}$, we denote by $(\mathfrak{Q})[\mathfrak{N}]$ the logic constructed by extending quantifier-free first-order logic with the quantifiers in $\mathfrak{Q}$ and allowing the numerical predicates in $\mathfrak{N}$.

   We denote by \emph{FO} the set of standard first-order quantifiers: $\{\exists, \forall\}$.
 \end{definition}
 When $\mathfrak{Q}$ is just a singleton $\{Q\}$, we sometimes denote $(\mathfrak{Q})[\mathfrak{N}]$ by $(Q)[\mathfrak{N}]$ or $(Q_1)[\mathfrak{N}]$ to emphasise that it is the unary quantifier.  We also write $\bar{Q}$ for the collection  of all vectorizations of the quantifier $Q$.
 We use similar notation for the sets of numerical predicates.  We use $\mathcal{L}((\mathfrak{Q})[\mathfrak{N}])$ to denote the languages expressible by the logic $(\mathfrak{Q})[\mathfrak{N}]$.  

All the logics we consider are \emph{substitution closed} in the sense of~\cite{Ebb85}.  This means in particular that if a quantifier $Q$ is definable in a logic $(\mathfrak{Q})[\mathfrak{N}]$, then extending the logic with the quantifier $Q$ does not add to its expressive power.  This is because we can replace occurrences of the quantifier $Q$ by its definition, with a suitable substitution of the interpretation for the relation symbols.  Hence, if $Q$ is definable in $(\mathfrak{Q})[\mathfrak{N}]$, then $\mathcal{L}((\mathfrak{Q})[\mathfrak{N}]) = \mathcal{L}((\mathfrak{Q} \cup \{Q\})[\mathfrak{N}])$.

A remark is due on our notation for numerical predicates.  All structures we consider are ordered, including those defining the quantifiers.  Thus the order predicate is implicitly present in the collection of numerical predicates $\mathfrak{N}$ and is used (implicitly) to define the interpretations to a quantifier.  We sometimes write $(\mathfrak{Q})[\varnothing]$ to indicate a logic in which this is the only use of the order that is allowed; by our choice of notation, the order symbol then does not appear explicitly in the syntax of the formulas.

\subsection{Multiplication Quantifiers}

The definition of multiplication quantifiers has its origin in Barrington, Immerman, and Straubing \cite[Section 5]{barrington1990uniformity} where they were referred to as monoid quantifiers; the authors proved that the languages in NC${}^1$ are exactly those expressible by first-order logic with quantifiers whose truth-value is determined via multiplication in a finite monoid.  The notion was extended by Lautemann et al.~\cite{lautemann2001descriptive} to include quantifiers for the word problem over more general algebras with a binary operation.  Multiplication quantifiers over a finite monoid $M$ can be understood as generalized quantifiers corresponding to languages recognized by $M$, and here we define them as such.  We then see how this definition matches that of multiplication quantifiers \textit{\`a la} Barrington et al.~\cite{barrington1990uniformity,krebs2007characterizing}.

Fix a monoid $M$, a set $B \subseteq M$, and a positive integer $k$.  Let $\Sigma_k$ denote the set $\{0,1\}^k$ which we think of as an alphabet of size $2^k$, and fix a function $\gamma: \Sigma_k \rightarrow M$.  We extend $\gamma$ to strings in $\Sigma_k^*$ inductively in the standard way: $\gamma(\epsilon) = 1_M$ and $\gamma(wa) = \gamma(w)\gamma(a)$.  Note that $\gamma$ is a monoid homomorphism.  Together these define a language
$$L^{M,B}_{\gamma} = \{ x \in \Sigma_k^* \mid \gamma(x) \in B\}.$$
We can now define a \emph{multiplication quantifier}.  In the following, $\mathcal{S}_L$ denotes the class of $\tau$ structures associated with a language $L$ in the sense of Definition~\ref{def:associated-lang}.
\begin{definition}\label{def:multquant}
Let $\tau$ be a vocabulary including an order symbol $<$ and $k$ unary relations.
For a monoid $M$, a set $B \subseteq M$ and a function $\gamma: \{0,1\}^k \rightarrow M$, the \emph{multiplication quantifier} $\Gamma_{\gamma}^{M,B}$ is the Lindstr\"om quantifier associated with the class of structures $\mathcal{S}_{L^{M,B}_{\gamma}}$.

We also write $\Gamma_{d,\gamma}^{M,B}$ for the vectorization of this quantifier of dimension $d$. If $B$ is a singleton $\{s\}$, then we may write $\Gamma^{M,s}_{d,\gamma}$ for short. 
\end{definition}

Recall that $U_1$ denotes the two-element monoid $\{0,1\}$ with standard multiplication.  Then, it is easily seen that $\Gamma^{U_1,0}_{1,\gamma}$,  where $\gamma : \{0,1\} \rightarrow U_1$ such that $\gamma(0) = 1$ and $\gamma(1) = 0$, is the standard existential quantifier.  The universal quantifier can be defined similarly.

Another way of describing the semantics of the multiplication quantifier (which relates it directly to the form described in Barrington et al.~\cite{barrington1990uniformity}) is as follows.  For a monoid $M$, a set $B \subseteq M$, a positive integer $k$ and a function $\gamma : \{0,1\}^k \rightarrow M$, consider the formula
$$
    \Gamma^{M,B}_{d,\gamma} \overline{x}\overline{y}(\phi_{<}(\overline{x}),\phi_1(\overline{y}), \dots, \varphi_k(\overline{y}))
    $$
    where $|\overline{x}| = 2d$ and $|\overline{y}| = d$; and  $I = (\phi_{<}(\overline{x}),\phi_1(\overline{y}), \dots, \varphi_k(\overline{y}))$ defines an interpretation of dimension $d$ from some vocabulary $\tau$ to a vocabulary with binary relation $\sqsubset$ and $k$ unary relations.
    Then, for a $\tau$-structure $\mathfrak{A}$ and assignment $\alpha$, we have that
    $$
    \mathfrak{A} \models \Gamma^{M,B}_{d,\gamma} \overline{x}\overline{y}(\phi_{<}(\overline{x}),\phi_1(\overline{y}), \dots, \varphi_k(\overline{y})) [\alpha]
    $$ if, and only if, $\phi_{<}$ defines a linear order $\sqsubset$ on the $d$-tuples of $\mathfrak{A}$ and
 $$
 \prod_{(\overline{a} \in ||\mathfrak{A}||^d)_{\sqsubset}} \gamma(\phi_1(\overline{a}), \dots, \phi_k(\overline{a})) \in B.
 $$
Here, multiplication is in the monoid $M$ and, since multiplication is not necessarily commutative, the order of $d$-tuples is specified to be the one given by $\sqsubset$. 

As stated above, when $\mathfrak{A}$ is itself ordered and the order $\sqsubset$ is the lexicographic order on $d$-tuples, we simply omit the formula $\phi_<$.

\begin{definition}
    For a monoid $M$, we define the following collections of quantifiers: 
    \begin{align*}
        \Gamma^M &= \left\{\Gamma^{M,B}_{d,\gamma} \mid B \subseteq M \text{, } \gamma : \{0,1\}^k \rightarrow M \text{, and } d,k \geq 1\right\}\\
        \Gamma^M_{d} &= \left\{\Gamma^{M,B}_{d,\gamma} \mid B \subseteq M \text{ and } \gamma : \{0,1\}^k \rightarrow M\right\}\\
        \Gamma^M_{\gamma} &= \left\{\Gamma^{M,B}_{d,\gamma} \mid B \subseteq M \text{, and } d \geq 1 \right\}\\
        \Gamma^M_{d,\gamma} &= \left\{\Gamma^{M,B}_{d,\gamma} \mid B \subseteq M\right\}
    \end{align*}
    Finally, let $\Gamma^{\fin}$ be the collection of all multiplication quantifiers over finite monoids.
\end{definition}

The study of the expressive power of multiplication quantifiers has generally been in relation to logics which restrict the application of multiplication quantifiers to interpretations with a lexicographic order.  For this purpose, we introduce a piece of notation: we write $\text{lex-}(\mathfrak{Q}  \cup \Gamma^{M})[\mathfrak{N}]$ to mean the fragment of the logic $(\mathfrak{Q}  \cup \Gamma^{M})[\mathfrak{N}]$ in which all applications of quantifiers from $\Gamma^{M}$ are to interpretations with a lexicographic order.  Likewise, we write $\text{fo-}(\mathfrak{Q}  \cup \Gamma^{M})[\mathfrak{N}]$ to mean the fragment of the logic $(\mathfrak{Q}  \cup \Gamma^{M})[\mathfrak{N}]$ in which all applications of quantifiers from $\Gamma^{M}$ are to interpretations in which the order is defined by a formula of $(\FO)[<]$.  In particular, there are no nested occurrences of the multiplication quantifier in the formula $\varphi_{<}$ defining the order.  Note also that, since the lexicographical order is definable by a first-order formula, $\text{lex-}(\mathfrak{Q}  \cup \Gamma^{M})[\mathfrak{N}]$ is contained in $\text{fo-}(\mathfrak{Q}  \cup \Gamma^{M})[\mathfrak{N}]$.

From \cite[Corollary 9.1]{barrington1990uniformity}, we know that $\NC^1$ is characterized by $(\text{FO})[+,\times]$ equipped with finite multiplication quantifiers over lexicographic orders:

\begin{theorem}[\cite{barrington1990uniformity}]\label{thm:logfornc1}
  $\NC^1 = \mathcal{L}(\emph{\text{lex-}}(\Gamma^{\fin})[+,\times])$.
\end{theorem}

\begin{remark}\label{thm:s5logfornc1}
    In fact, simply adding multiplication quantifiers for some fixed finite, non-solvable monoid to $(\FO)[+,\times]$ suffices.  The definition of ``non-solvable monoid'' is not needed for our proofs here but, for example, the \emph{symmetric group of degree five}, denoted $S_5$, is a non-solvable monoid. Therefore, we know that $\NC^1 = \mathcal{L}({\text{lex-}}(\FO\cup\Gamma^{S_5})[+,\times])$.
\end{remark}

In the absence of the arithmetic predicates for addition and multiplication, the logic of multiplication quantifiers over finite monoids only allows us to define regular languages.  Specifically, Barrington et al.~\cite[Theorem 11.1]{barrington1990uniformity} established that the regular languages are characterized by the logic using such quantifiers with only unary interpretations.
\begin{theorem}[\cite{barrington1990uniformity}]\label{thm:fologforreg}
    \emph{\textsc{Reg}} $= \mathcal{L}(\emph{\text{lex-}}(\Gamma^{\fin}_1)[<])$.
\end{theorem}
\noindent Later, Lautemann et al.~\cite[Theorem 5.1]{lautemann2001descriptive} showed that allowing interpretations of higher dimension to the quantifiers does not increase the expressive power when order is the only numerical predicate.
\begin{theorem}[\cite{lautemann2001descriptive}]\label{thm:unaryfologforreg}
    \emph{\textsc{Reg}} $= \mathcal{L}(\emph{\text{lex-}}(\Gamma^{\fin})[<])$.
  \end{theorem}

In Theorem~\ref{thm:finitebinding} we show that this is true even in the presence of other numerical predicates and, therefore, $\Gamma^{\fin}$ can be replaced by $\Gamma^{\fin}_1$ even in Theorem~\ref{thm:logfornc1}.  For our intended application, we need this technical result only in the case when the quantifier in $\Gamma^{\fin}$ is applied to interpretations where the order defined is lexicographic, and we first give the proof in this special case.  However, the result can be extended to cases in which the quantifier in $\Gamma^\fin$ is applied to interpretations where the order is defined by a $(\FO)[<]$-formula, and this may be of independent interest so we state the more general Theorem~\ref{thm:morethanlex}.

\subsection{Typed Monoids}\label{sec:backgroundtyped1}

Our results build on the theory of \emph{typed monoids} as developed in the work of Krebs and collaborators~\cite{behle2007linear,behle2011typed,krebs2007characterizing,krebs2008typed,cano2021positive}.  We recapitulate the main definitions and results from~\cite{behle2011typed,krebs2008typed} on typed monoids, their relationship to languages, and corresponding characterizations of complexity classes here.

A typed monoid is a monoid equipped with a collection of \emph{types}, which form a Boolean algebra, and a set of \emph{units}.  We only deal with concrete Boolean algebras, given as collections of subsets of a fixed universe: a \emph{Boolean algebra} over a set $S$ is a set $B \subseteq \wp(S)$ such that $\varnothing, S \in B$ and $B$ is closed under union, intersection, and complementation. If $B$ is finite, we call it a \emph{finite} Boolean algebra. We call $\varnothing$ and $S$ the \emph{trivial elements} (or in some contexts, the \emph{trivial types}) of $B$.

A homomorphism between Boolean algebras is defined as standard.  That is, if $B_1$ and $B_2$ are Boolean algebras over sets $S$ and $T$, respectively, then we call $h : B_1 \rightarrow B_2$ a \emph{homomorphism} if $h(\varnothing) = \varnothing$, $h(S) = T$, and for all $s_1, s_2 \in B_1$, $h(s_1 \cap s_2) = h(s_1) \cap h(s_2)$, $h(s_1 \cup s_2) = h(s_1) \cup h(s_2)$, and $h(s^C) = (h(s))^C$.
Now we are ready to define typed monoids.

Let $M$ be a monoid, $G$ a Boolean algebra over $M$, and $E$ a finite subset of $M$. We call the tuple $T = (M, G, E)$ a \emph{typed monoid over $M$} and the elements of $G$ \emph{types} and the elements of $E$ \emph{units}. We call $M$ the \emph{base monoid} of $T$. If $M$ is a group, then we may also call $T$ a \emph{typed group}. When $G = \{\varnothing, A, M - A, M\}$ for some $A \subseteq M$, we abbreviate $T$ as $(M, A, E)$, i.e., the Boolean algebra is signified by an element, or elements, which generates it---in this case, $A$.  We say that a typed monoid $(M,G,E)$ is finite if $M$ is.

We also need a notion of morphism between typed monoids. 
    A \emph{typed monoid homomorphism} $h$ from $(S, G, E)$ to  $(T, H, F)$ is a triple $(h_1, h_2, h_3)$ where $h_1 : S \rightarrow T$ is a monoid homomorphism, $h_2 : G \rightarrow H$ is a homomorphism of Boolean algebras, and $h_3 : E \rightarrow F$ is a mapping of sets such that the following conditions hold:
    \begin{enumerate}[\ \ \ \ (i)]
        \item For all $A \in G$, $h_1(A) = h_2(A) \cap h_1(S)$.
        \item For all $e \in E$, $h_1(e) = h_3(e)$.
    \end{enumerate} 

Note that $h_3$ is redundant in the definition as it is completely determined by $h_1$.  We retain it as part of the definition for consistency with~\cite{behle2011typed,krebs2008typed}.

To motivate the definitions, recall that a language $L \subseteq \Sigma^*$ is recognized by a monoid $M$ if there is a homomorphism $h: \Sigma^* \rightarrow M$ and a set $B \subseteq M$ such that $L = h^{-1}(B)$.  When the monoid $M$ is infinite, the languages recognized form a rather rich collection and we aim to restrict this in two ways.  First, $B$ cannot be an arbitrary set but must be an element of the algebra of types.  Secondly, the homomorphism $h$ must map the letters in $\Sigma$ to units of the typed monoid. Formally, we write that a typed monoid $T = (M, G, E)$ \emph{recognizes} a language $L \subseteq \Sigma^*$ if there exists a typed monoid homomorphism from $(\Sigma^*, L, \Sigma)$ to $T$. We let $\mathcal{L}(T)$ denote the set of languages recognized by $T$. When the base monoid of a typed monoid is finite, we recover the classical definition of recognition.  Hence, the languages recognized by finite typed monoids are necessarily regular.

We can now state the definitions of the key relationships between typed monoids. Let $(S, G, E)$ and $(T, H, F)$ be typed monoids.
    \begin{itemize}
        \item A typed monoid homomorphism $h = (h_1, h_2, h_3) : (S, G, E) \rightarrow (T, H, F)$ is \emph{injective} (\emph{surjective}, or \emph{bijective}) if all of $h_1$, $h_2$, and $h_3$ are.
        \item $(S, G, E)$ is a \emph{typed submonoid} (or, simply, ``submonoid'' when the meaning is obvious from context) of $(T, H, F)$, denoted $(S, G, E) \leq (T, H, F)$, if there exists an injective typed monoid homomorphism $h : (S, G, E) \rightarrow (T, H, F)$.
        \item $(S, G, E)$ divides $(T, H, F)$, denoted $(S, G, E) \preceq (T, H, F)$, if there exists a surjective  typed monoid homomorphism from a submonoid of $(T, H, F)$ to $(S, G, E)$.
    \end{itemize}

\noindent These have the expected properties. Let $T_1$, $T_2$, and $T_3$ be typed monoids:
    \begin{itemize}
        \item Typed monoid homomorphisms are closed under composition.
        \item Division is transitive: if $T_1 \preceq T_2$ and $T_2 \preceq T_3$, then $T_1 \preceq T_3$.
        \item If $T_1 \preceq T_2$, then $\mathcal{L}(T_1) \subseteq \mathcal{L}(T_2)$.
    \end{itemize}

We can formulate the notion of the \emph{syntactic typed monoid} of a language $L$ as an extension of the syntactic monoid of $L$ with a minimal collection of types and units necessary. Let $T = (M, G, E)$ be a typed monoid. A congruence $\sim$ over $M$ is a \emph{typed congruence over $T$} if for every $A \in G$ and $s_1,s_2\in M$, if $s_1 \sim s_2$ and $s_1 \in A$, then $s_2 \in A$. For a typed congruence $\sim$ over $T$, let
    \begin{align*}
        A/{\sim} &= \{[x]_\sim \mid x \in A\} \text{ where $A \subseteq M$}\\
        G/{\sim} &= \{A/{\sim} \mid A \in G\}\\
        E/{\sim} &= \{[x]_\sim \mid x \in E\}.
    \end{align*}
    \noindent Then, $T/{\sim} := (M/{\sim}, G/{\sim}, E/{\sim})$ is the \emph{typed quotient monoid of $T$ by $\sim$}.

Let $\sim_T$ denote the typed congruence on $T$ such that for $s_1,s_2 \in S$, $s_1 \sim_T s_2$ if, and only if, for all $x,y \in S$ and $A \in G$, $xs_1y \in A$ if, and only if,  $xs_2y \in A$. We then refer to the quotient monoid $T/{\sim_T}$ as the \emph{minimal reduced monoid of $T$}.

Recall that ${\sim_L}$ is the syntactic congruence of $L$, defined in Definition \ref{def:syncong}. For a language $L \subseteq \Sigma^*$, the \emph{syntactic typed monoid of $L$}, denoted $\syn(L)$, is the typed monoid $(\Sigma^*, L, \Sigma)/{\sim_L}$. 

    We also get the \emph{canonical typed monoid homomorphism}, $\eta_L : (\Sigma^*, L, \Sigma) \rightarrow \syn(L)$ induced by the syntactic homomorphism of $L$.

It also turns out that we can give a purely structural characterization of those typed monoids that are syntactic monoids.

\begin{proposition}[\cite{krebs2008typed}]
    A typed monoid is the syntactic monoid of a language if, and only if, it is reduced, generated by its units, and has four or two types.

In case it has just two types, then it only recognizes the empty language or the language of all strings.
\end{proposition}

\subsection{Typed Monoids and Multiplication Quantifiers}\label{sec:backgroundtyped2}
Here, we turn to discussing the relationship between the expressive power of logics with multiplication quantifiers and typed monoids.  A formal association is defined through the definition below.

\begin{definition}\label{def:origtypedquantsemi}
    For a multiplication quantifier $Q=\Gamma^{M,B}_{\gamma}$ where $\gamma : \{0,1\}^k \rightarrow M$, we define the \emph{typed quantifier monoid} of $Q$ to be the syntactic typed monoid of the language $L^{M,B}_{\gamma}$.  
\end{definition}

We also state the formal connection between the languages expressible in a logic with a collection of quantifiers and the corresponding class of typed monoids which recognize the languages.  For this we need the notion of the \emph{ordered strong block product} of a pair of monoids.  The definition is technical and can be found in Krebs~\cite{krebs2008typed} and is also reproduced in Section~\ref{sec:bp}.  For a set of typed monoids $T$, we denote by  $\mathrm{sbpc}_<(T)$ its closure under ordered strong block products.

From Krebs \cite[Theorem 4.14]{krebs2008typed}, we then get the following relationship between logics and algebras:\footnote{The theorem in \cite{krebs2008typed} is actually more general as it allows for more predicates than just order; however, for our purposes, order alone suffices.}
\begin{theorem}\label{thm:logcircalgequiv}
     Let $\mathfrak{Q}$ be a collection of multiplication quantifiers and $\bm{Q}$ the set of typed quantifier monoids for quantifiers in $\mathfrak{Q}$. Then, \(
        \mathcal{L}(\emph{\text{lex-}}(\mathfrak{Q}_1)[<]) = \mathcal{L}(\emph{sbpc}_<(\bm{Q}))
    \). 
  \end{theorem}
Recall that the subscript in $\mathfrak{Q}_1$ means we are restricting the logic to only use the unary quantifiers in $\mathfrak{Q}$ and the prefix ``lex-'' means that the quantifiers are only applied to interpretations with a lexicographic order, which in the context of unary quantifiers just means that the order defined is $<$ itself.

\subsection{Strong Block Product Closure}\label{sec:bp}

Here, we build up to the definition of the \emph{ordered strong block product} of a pair of typed monoids.

\subsubsection{Weakly Closed Classes}
  
The \emph{direct product of two monoids} $(S, \cdot_S)$ and $(T, \cdot_T)$ is the monoid $(S \times T, \cdot)$ where $(s_1, t_1) \cdot (s_2, t_2) = (s_1 \cdot_S s_2, t_1 \cdot_T t_2)$. We define the \emph{direct product of Boolean algebras} $B_1$ and $B_2$, denoted $B_1 \times B_2$, to be the Boolean algebra generated by the set $\{A_1 \times A_2 \mid A_1 \in B_1 \text{ and } A_2 \in B_2\}$. Finally, The \emph{direct product of typed monoids} $(S, G, E) \times (T, H, F)$ is the typed monoid $(S \times T, G \times H, E \times F)$.
      
If there exists a surjective typed monoid homomorphism from $(S, G, E)$ to $(T, H, F)$, then we say that $(S, G, E)$ is a \emph{trivial extension} of $(T, H, F)$.
  
We call a set of typed monoids $T$ a \emph{weakly closed class} if it is closed under
\begin{itemize}
    \item Division: If $(S, G, E) \in T$ and $(S, G, E) \preceq (T, H, F)$, then $(T, H, F) \in T$.
    \item Direct Product: If $(S, G, E), (T, H, F) \in T$, then $(S, G, E) \times (T, H, F) \in T$.
    \item Trivial Extension: If $(S, G, E)$ is a trivial extension of $(T, H, F)$ and $(T, H, F) \in T$, then $(S, G, E) \in T$.
\end{itemize}
We write $\text{wc}(T)$ to denote the smallest weakly closed set of typed monoids containing $T$.
  
\subsubsection{The Block Product}
The block product is our main tool for the construction of algebraic characterizations of language classes via logic. Historically, the ``\emph{wreath product}'' was first used for this purpose. Since Rhodes and Tilson \cite{rhodes1989kernel}, however, the block product has typically been the preferred and easier-to-work-with tool of choice. We now build up to its definition:

A \emph{left action} $\star_l$ of a monoid $(N, \cdot)$ on a monoid $(M, +)$ is a function from $N \times M$ to $M$ such that for $n,n_1,n_2 \in N$ and $m,m_1,m_2 \in M$, \begin{align*}
    n \star_l (m_1 + m_2) &= n \star_l m_1 + n \star_l m_2\\
    (n_1 \cdot n_2) \star_l m &= n_1 \star_l (n_2 \star_l m)\\
    n \star_l 1_M &= 1_M\\
    1_N \star_l m &= m
\end{align*}
The \emph{right action} $\star_r$ of $(N, \cdot)$ on $(M, +)$ is defined dually. We say that left and right actions of $(N, \cdot)$ on $(M,+)$ are \emph{compatible} if for all $n_1,n_2 \in N$ and $m \in M$, \[
    (n_1 \star_l m) \star_r n_2 = n_1 \star_l (m \star_r n_2).
\] When clear from context, we may simply write $nm$ for $n \star_l m$ and $mn$ for $m \star_r n$.

For a pair of compatible left and right actions, $\star_l$ and $\star_r$ of $(N, \cdot)$ on $(M, +)$, the \emph{two-sided (or bilateral) semidirect product} of $(M, +)$ and $(N, \cdot)$ with respect to $\star_l$ and $\star_r$ is the monoid $(M \times N, \star)$ where for $(m_1,n_1),(m_2,n_2) \in M \times N$, \[
    (m_1,n_1) \star (m_2,n_2) = (m_1n_2 + n_1m_2, n_1 \cdot n_2).
\]

The \emph{block product} of $(M, \cdot_M)$ with $(N, \cdot_N)$, denoted $M \Box N$, is the two-sided semidirect product of $(M^{N \times N}, +)$ and $(N, \cdot)$ with respect to the left and right actions $\star_l$ and $\star_r$ where for $f,g \in M^{N \times N}$ and $n,n_1,n_2 \in N^1$,
\begin{itemize}
    \item $(M^{N \times N}, +)$ is the monoid of all functions from $N \times N$ to $M$ under componentwise product $+$: \[
        (f + g)(n_1,n_2) = f(n_1,n_2) \cdot_M g(n_1,n_2).
    \]
    \item The left action $\star_l$ of $(N, \cdot)$ on $(M^{N \times N}, +)$ is defined by \[
        (n \star_l f)(n_1, n_2) = f(n_1 \cdot_N n, n_2).
    \]
    \item The right action $\star_r$ of $(N, \cdot)$ on $(M^{N \times N}, +)$ is defined by \[
        (f \star_r n)(n_1, n_2) = f(n_1, n \cdot_N n_2).
    \]
\end{itemize}

\subsubsection{The Typed Block Product}

Let $(S, G, E)$ and $(S', G', E')$ be typed monoids and $C \subseteq S'$ be a finite set. Then, the \emph{typed block product with $C$} of $(S, G, E)$ and $(S', G', E')$, denoted $(S, G, E) \boxdot_C (S', G', E')$, is the typed monoid $(T, H, F)$ where
\begin{enumerate}[\ \ \ \ (1)]
    \item $T \leq S \Box S'$ such that $T$ is generated by the elements $(f,s')$ such that
    \begin{enumerate}[\ \ \ \ (a)]
        \item $s' \in E' \cup C$ and
        \item $f \in E^{S' \times S'}$ such that for $b_1,b_2,b_3,b_4 \in S'$, if for all $c \in C$ and all $A' \in G'$, \(
            b_1cb_2 \in A' \text{ iff } b_3cb_4 \in A',
        \) then $f(b_1,b_2) = f(b_3,b_4)$,
    \end{enumerate}
    \item $H = \{\{(f,s) \mid f(1,1) \in A\} \mid A \in G\}$ where $1$ is the identity of $S'$,
    \item and $F = \{(f,s') \mid \text{$(f,s)$ is a generator of $T$ and } s' \in E'\}$.
\end{enumerate}
    
Because the typed monoid corresponding to the order predicate is frequently needed, it is convenient to define an \emph{ordered typed block product}, $(S, G, E) \boxtimes_C (S', G', E')$ which helps simplify our algebraic representations whose numerical predicates only include order; this is defined like the typed block product above but with a change to condition (1)(b):
\begin{enumerate}[\ \ \ \ {(1)}(b${}_<$)]
    \setcounter{enumi}{1}
    \item $f \in E^{S' \times S'}$ such that for $b_1,b_2,b_3,b_4 \in S'$, if for all $c \in C$ and all $A' \in G'$,
    \begin{enumerate}[\ \ \ \ (i)]
        \item $b_1cb_2 \in A'$ iff $b_3cb_4 \in A'$,
        \item $b_1c \in A'$ iff $b_3c \in A'$,
        \item and $cb_2 \in A'$ iff $cb_4 \in A'$,
    \end{enumerate}
    then $f(b_1,b_2) = f(b_3,b_4)$.
\end{enumerate}

For a set of typed monoids $W$, we let \[
    W_0 = \emph{wc}(W)
\] and for each $k \geq 1$,
\begin{itemize}
    \item $W_k = \{S_1 \boxdot_C S_2 \mid S_1 \in W_0 \text{, } S_2 \in W_{k-1} \text{, and finite } C \subseteq S_2\}$
    \item $W_k^< = \{S_1 \boxtimes_C S_2 \mid S_1 \in W_0 \text{, } S_2 \in W_{k-1}^< \text{, and finite } C \subseteq S_2\}$
\end{itemize}
We define the \emph{(ordered) strong block product closure of $W$}, denoted $\text{sbpc}(W)$ ($\text{sbpc}_<(W)$), as
\begin{itemize}
    \item $\text{sbpc}(W) = \bigcup_{k \in \mathbb{N}} W_k$
    \item $\text{sbpc}_<(W) = \bigcup_{k \in \mathbb{N}} W_k^<$.
\end{itemize}

\section{Simplifying Multiplication Quantifiers}\label{sec:mult}

To use Theorem~\ref{thm:logcircalgequiv} to obtain an algebraic characterization of $\NC^1$, we need to characterize this class in a logic with only unary quantifiers.  Remark~\ref{thm:s5logfornc1} gives us a characterization using first-order quantifiers and quantifiers in $\Gamma^{S_5}$.  Our aim in this section is to show that we can eliminate the use of quantifiers of dimension higher than $1$ in this logic.  As a first step, we show that we can restrict ourselves to quantifiers $\Gamma^{{S_5},B}_{\delta}$ for a \emph{fixed} function $\delta$. 

\begin{lemma}\label{lem:finitetuple}
    For every finite monoid $M$, there exists a function $\delta : \{0,1\}^{|M|} \rightarrow M$ such that for every $B \subseteq M$ and $\gamma : \{0,1\}^k \rightarrow M$, and dimension $d$, the quantifier $\Gamma^{M,B}_{d,\gamma}$ is definable in $(\Gamma^{M,B}_{d,\delta})[<]$.
  \end{lemma}
  \begin{proof}
    Recall that $\Gamma^{M,B}_{d,\gamma}$ is the class of structures $\str{A}$ in a vocabulary $\tau$ with one $2d$-ary ordering relation $<$ and $k$ $d$-ary relations $R_1,\ldots,R_k$ such that $\gamma(w_{\str{A}}) \in B$ where $w_{\str{A}}$ is the $||\mathfrak{A}||^d$-length string associated with $\str{A}$ as in Def.~\ref{def:associated}.

    Let $c = |M|$, fix an enumeration $\{m_1,\ldots,m_c\}$ of $M$, and let $z$ be an arbitrary element of $M$.  Let $\delta : \{0,1\}^c \rightarrow M$ be the function where $\delta(w) = m_i$ if $w$ is the one-hot encoding of $i$ and $\delta(w) = z$, for some arbitrary $z \in M$, otherwise (that is, if the number of occurrences of the symbol $1$ in the string $w$ is not exactly one).

    For each $t \in M$, define the formula $\psi_t(y_1,\dots,y_d)$ as follows:
    \begin{align*}
        \psi_t(y_1,&\dots,y_d) : = \\
        &\bigvee_{w \in \{0,1\}^k : \gamma(w) = t} \left( \bigwedge_{i \in [k] : w_i = 1} R_i(y_1,\dots,y_d) \land \bigwedge_{i \in [k] : w_i = 0} \neg R_i(y_1,\dots,y_d) \right).
    \end{align*}

    It is easy to see that in a $\tau$-structure $\str{A}$, we have $\str{A} \models \psi_t[a_1/y_1, \dots, a_d/y_d]$ if, and only if, the element of $w_{\str{A}}$ indexed by $(a_1, \dots, a_d)$ is mapped by $\gamma$ to $t$.  Thus, in particular, the formulas $\psi_{m_1},\ldots,\psi_{m_c}$ define disjoint sets that partition the universe of $\str{A}$.  We now claim that the quantifier $\Gamma^{M,B}_{\gamma}$ is defined by the formula:
      $$ \Gamma^{M,B}_{\delta}\overline{x_1}\overline{x_2}y_1\dots y_d(\overline{x_1}<\overline{x_2}, \psi_{m_1}(y_1,\dots,y_d),\ldots,\psi_{m_c}(y_1,\dots,y_d)).$$

      To see this, let $I$ denote the interpretation $(<,\psi_{m_1}\ldots,\psi_{m_c})$ so that $w_{I(\str{A})}$ is a string over $\{0,1\}^c$.  By the fact that the sets defined by the formulas $\psi_{m_1},\ldots,\psi_{m_c}$ partition $|\str{A}|$ it follows that each letter of $w_{I(\str{A})}$ is a vector in $\{0,1\}^c$ with exactly one $1$.  Indeed, the element of $w_{I(\str{A})}$ indexed by $(a_1,\dots,a_d)$ is the one-hot encoding of $i$ precisely if $\str{A} \models \psi_{m_i}[a_1/y_1,\dots,a_d/y_d]$.  Since $\delta$ takes the one-hot encoding of $i$ to $m_i$, we have for any $a_1,\dots,a_d \in |\str{A}|$
      \begin{align*}
        & \delta((w_{I(\str{A})})_{(a_1,\dots,a_d)})  =  m_i \\
        \text{iff}\ & \str{A} \models\psi_i[a_1/y_1, \dots, a_d/y_d] \\
        \text{iff}\ & \gamma((w_{\str{A}})_{(a_1,\dots,a_d)}) = m_i.
      \end{align*}
        Hence, $\delta(w_{I(\str{A})}) = \gamma(w_{\str{A}})$ and therefore $I(\str{A}) \in \Gamma^{M,B}_{d,\delta}$ if, and only if, $\str{A} \in \Gamma^{M,B}_{d,\gamma}$ as required.
        
  \end{proof}

  It then follows from Lemma~\ref{lem:finitetuple} and the substitution property of quantifiers that the expressive power of $(\mathfrak{Q} \cup \Gamma^{M})[\mathfrak{N}]$ is the same as that of $(\mathfrak{Q} \cup \Gamma^{M}_{\delta})[\mathfrak{N}]$.  Indeed, any application of a quantifier in $\Gamma^M$ can be replaced by an application of a quantifier in $\Gamma^{M}_{\delta}$ of the same dimension.  We next aim to show that an application of a quantifier $\Gamma^{M,B}_{d,\delta}$ with lexicographic order, can be replaced by $d$ nested applications of quantifiers $\Gamma^{M,B}_{1,\delta}$

  \begin{lemma}\label{lem:nesting}
    For any collection of quantifiers $\mathfrak{Q}$  and numerical predicates $\mathfrak{N}$, any formula of $\emph{\text{lex-}}(\mathfrak{Q}  \cup \Gamma^{M}_{\delta})[\mathfrak{N}]$ is equivalent to one of $\emph{\text{lex-}}(\mathfrak{Q}  \cup \Gamma^{M}_{1,\delta})[\mathfrak{N}]$.
  \end{lemma}
  \begin{proof}
Again, fix an enumeration $M = \{m_1,\ldots,m_c\}$ of $M$ and recall that $\delta : \{0,1\}^c \rightarrow M$ takes the one-hot encoding of $i$ to $m_i$.
    
    We show, by induction on $d$ that if we have a formula
    \[ \Phi := \Gamma^{M,B}_{d,\delta} x_1,\ldots,x_d (\phi_1(x_1,\ldots,x_d),\ldots,\phi_c(x_1,\ldots,x_d)) \]
 where each formula $\phi_i$ is in $\text{lex-}(\mathfrak{Q}  \cup \Gamma^{M}_{1,\delta})[\mathfrak{N}]$, then $\Phi$  is equivalent to a formula of $\text{lex-}(\mathfrak{Q}  \cup \Gamma^{M}_{1,\delta})[\mathfrak{N}]$.  The result then follows by induction on the structure of the formula.

    The base case when $d=1$ is trivially true.  Assume then that we have the formula $\Phi$ for $d > 1$ and let $$I = (\phi_1(x_1,\ldots,x_d),\ldots,\phi_c(x_1,\ldots,x_d))$$ denote the interpretation of dimension $d$.

    We claim that $\Phi$ is equivalent to the formula
    \[ \Phi_1 := \Gamma_{1,\delta}^{M,B} x_1  (\theta_1(x_1),\ldots,\theta_c(x_1)) \]
    where $\theta_i$ is the formula
    \[ \theta_i(x_i) := \Gamma_{d-1,\delta}^{M,m_i}x_2,\ldots,x_d (\phi_1(x_1,\ldots,x_d),\ldots,\phi_c(x_1,\ldots,x_d)).\]

    Thus, $\Phi_1$ is obtained by the application of $ \Gamma_{\delta}^{M,B}$ to an interpretation $$I_1 := (\theta_1(x_1),\ldots,\theta_c(x_1))$$ where each formula $\theta_i$ is obtained as the application of a quantifier $\Gamma_{d-1,\delta}^{M,s}$ to an interpretation defined by formulas of $\text{lex-}(\mathfrak{Q} \cup \Gamma^{M}_{1,\delta})[\mathfrak{N}]$.  Thus, by the inductive hypothesis, each $\theta_i$ is equivalent to a formula of $\text{lex-}(\mathfrak{Q}  \cup \Gamma^{M}_{1,\delta})[\mathfrak{N}]$ and we are done.

      It remains to show that $\Phi$ and $\Phi_1$ are equivalent on any structure $\str{A}$.  To see this, fix an assignment $\alpha$ of values in $|\str{A}|$ to the free variables of $\Phi$.  We need to show that $\str{A} \models \Phi[\alpha]$ if, and only if, $\str{A} \models \Phi_1[\alpha]$.  Let $n$ be the length of $\str{A}$ and assume without loss of generality that the elements of $|\str{A}|$ are $\{1,\ldots,n\}$ in that order.

      Now, $w_{I(\str{A},\alpha)}$ denotes the string associated with the structure $I(\str{A},\alpha)$ and note that this is a string of length $n^d$ whose elements are indexed by $d$-tuples of elements of $\str{A}$ in lexicographic order.  By definition, $\str{A} \models \Phi[\alpha]$ precisely if $\delta(w_{I(\str{A},\alpha)}) \in B$.  We can also regard $I$ as an interpretation of dimension $d-1$ obtained by treating the variable $x_1$ as a parameter.  We write $w_{I(\str{A},\alpha[a/x_1])}$ for the string of length $n^{d-1}$ that results from applying this interpretation with the assignment of $a$ to the variable $x_1$.  Since the ordering of $d$-tuples in $w_{I(\str{A},\alpha)}$ is lexicographic, we have
\[
        w_{I(\str{A},\alpha)} = w_{I(\str{A},\alpha[1/x_1])}\cdots w_{I(\str{A},\alpha[n/x_1])}
\]
      and thus
\[
        \delta(w_{I(\str{A},\alpha)}) = \delta(w_{I(\str{A},\alpha[1/x_1])})\cdots \delta(w_{I(\str{A},\alpha[n/x_1])}).
\]

      Now, by definition of $\theta_i$, we have that $\str{A} \models \theta_i[\alpha[a/x_1]]$ if, and only if, $\delta(w_{I(\str{A},\alpha[a/x_1])}) = m_i$.  Thus, for each $a \in |\str{A}|$, there is exactly one $i$ such that  $\str{A} \models \theta_i[\alpha[a/x_1]]$.  Thus, 
      $w_{I_1(\str{A},\alpha)}$ is the string of length $n$ whose  $a$th element is the one-hot encoding of $i$ exactly when $\delta(w_{I(\str{A},\alpha[a/x_1])}) = m_i$.  In other words, $\delta((w_{I_1(\str{A},\alpha)})_a) = \delta(w_{I(\str{A},\alpha[a/x_1])})$.
      Thus,
    \begin{align*}
        & \str{A} \models \Phi[\alpha] \\
        \text{iff}\ & \delta(w_{I(\str{A},\alpha)}) \in B \\
        \text{iff}\ & \delta(w_{I(\str{A},\alpha[1/x_1])})\cdots \delta(w_{I(\str{A},\alpha[n/x_1])}) \in B \\
        \text{iff}\ & \delta((w_{I_1(\str{A},\alpha)})_1)\cdots \delta((w_{I_1(\str{A},\alpha)})_n) \in B \\
        \text{iff}\ & \delta(w_{I_1(\str{A},\alpha)}) \in B \\
        \text{iff}\ & \str{A} \models \Phi_1[\alpha].
    \end{align*}

    \end{proof}

    Now we are ready to state the main theorem of this section.
    \begin{theorem}\label{thm:finitebinding}
         For every finite monoid $M$, there exists a function $\delta : \{0,1\}^{|M|} \rightarrow M$ such that for any collection of quantifiers $\mathfrak{Q}$ and any set $\mathfrak{N}$ of numerical predicates, every formula of $\emph{\text{lex-}}(\mathfrak{Q} \cup \Gamma^M)[\mathfrak{N}]$ is equivalent to a formula of $\emph{\text{lex-}}(\mathfrak{Q} \cup \Gamma^M_{1,\delta})[\mathfrak{N}]$.
   \end{theorem}
   \begin{proof}
     Let $\phi$ be any formula of $\text{lex-} (\mathfrak{Q}  \cup \Gamma^M)[\mathfrak{N}]$.  By Lemma~\ref{lem:finitetuple} we can replace all occurrences of quantifiers in $\Gamma^{M,B}_{\gamma}$ by their definitions using quantifiers in $\Gamma^{M,B}_{\delta}$ to get a  formula $\phi'$ of $\text{lex-} (\mathfrak{Q}  \cup \Gamma^M_{\delta})[\mathfrak{N}]$ equivalent to $\phi$.  Finally, by Lemma~\ref{lem:nesting}, there is a formula of $\text{lex-}(\mathfrak{Q}  \cup \Gamma^M_{1,\delta})[\mathfrak{N}]$ equivalent to $\phi'$.
\end{proof}

Note that for a finite monoid $M$, while $\Gamma^{M}$ and $\Gamma^{M}_1$ are infinite sets, $\Gamma^{M}_{1,\delta}$ is a finite set.
Therefore, this gives us a logic characterizing $\NC^1$ which not only uses unary quantifiers but also only has a finite number of quantifiers:
\begin{corollary}\label{cor:alogtimelogicbetter}
    There exists a $\delta : \{0,1\}^k \rightarrow S_5$ such that \(
        \NC^1 = \mathcal{L}(\emph{\text{lex-}}(\FO \cup \Gamma^{S_5}_{1,\delta})[+,\times]).
    \)
\end{corollary}

\noindent This simplifies our construction of an algebra capturing $\NC^1$.

Another consequence of Theorem~\ref{thm:finitebinding} is Theorem~\ref{thm:unaryfologforreg} (\cite[Theorem~5.1]{lautemann2001descriptive}). Thus, we get a proof of Theorem~\ref{thm:unaryfologforreg} which is purely logical,  unlike the original proof which relies on the use of finite automata.  Furthermore, we also resolve a question left open in Lautemann et al.~\cite{lautemann2001descriptive}:
\begin{corollary}\label{cor:openquestionfromlautemann}
    $\mathcal{L}(\emph{\text{lex-}}(\Gamma^{\fin})[+,\times]) = \mathcal{L}(\emph{\text{lex-}}(\Gamma^{\fin}_1)[+,\times])$
\end{corollary}

\section{The Algebraic Characterization}\label{sec:alg}

Now that we have an extension of first-order logic capturing $\NC^1$ using only unary quantifiers, we are able to apply Theorem~\ref{thm:logcircalgequiv} to construct an algebra for it.  We just need to
obtain an equivalent logic using only the numerical predicate $<$ without introducing quantifiers of higher dimension.

To do this, we follow the construction of an algebra for $\TC^0$.  The \emph{majority} quantifier $\text{Maj}$ is the collection of strings over the alphabet $\{0,1\}$ in which at least half of the symbols are $1$.  The \emph{square} quantifier $\text{Sq}$ is the collection of strings over the alphabet $\{0,1\}$ in which the number of $1$s is a positive square number (i.e.,\ an element of $\mathbb{S}$).  In the logics we define below, we always use these quantifiers only with unary interpretations.

The following lemma displays some known results about the expressiveness of these quantifiers:
\begin{lemma}\label{lem:superhelperalogtime}
    \,
    \begin{enumerate}[\ \ \ \ (i)]
        \item \emph{Maj} is definable in $\emph{\text{lex-}}(\Gamma^\fin)[+,\times]$. (cf. Barrington et al. \cite{barrington1990uniformity})
        \item The quantifiers in \emph{FO} are definable in $(\emph{Maj})[<]$. (cf. Lange \cite[Theorem 3.2]{lange2004some})
        \item The numerical predicate $+$ is definable in $(\emph{Maj})[<]$. (cf. Lange \cite[Theorem 4.1]{lange2004some})
        \item The numerical predicate $\times$ is definable in $(\{\emph{Maj}, \emph{Sq}\})[<]$ and \emph{Sq} is definable in $(\emph{Maj})[<,+,\times]$. (cf. Schweikardt \cite[Theorem 2.3.f]{schweikardt2002expressive} and Krebs et al. \cite[Section 2.3]{krebs2007characterizing})
    \end{enumerate}
\end{lemma}

Bringing everything together, we get the following algebraic characterization of $\NC^1$:

\begin{theorem}\label{thm:algebraforalogtime}
    \[\NC^1 = \mathcal{L}(\emph{sbpc}_<(\{(\mathbb{Z}, \mathbb{Z}^+, \pm 1), (\mathbb{N}, \mathbb{S}, \{0,1\}), (S_5, \wp(S_5), S_5)\})).\]
\end{theorem}
\begin{proof}
    Let $\delta : \{0,1\}^c \rightarrow S_5$ be as defined in Lemma \ref{lem:finitetuple}. It is easy to see that the typed quantifier monoid for Maj is $(\mathbb{Z}, \mathbb{Z}^+, \pm 1)$, for Sq is $(\mathbb{N}, \mathbb{S}, \{0,1\})$, and for $\Gamma^{S_5,s}_{1,\delta}$ is $(S_5, \{s\}, S_5)$ for any $s \in S_5$.

    By Corollary \ref{cor:alogtimelogicbetter}, we have that \(
        \NC^1 = \mathcal{L}(\text{lex-}(\text{FO} \cup \Gamma^{S_5}_{1,\delta})[+,\times]).
    \) Lemma \ref{lem:superhelperalogtime} (ii)--(iv) allow us to define FO, $+$, and $\times$ in $(\{\text{Maj}, \text{Sq}\})[<]$ and (i) and (iv) allow us to define Maj and Sq in $\text{lex-}(\FO \cup \Gamma^{S_5}_{1,\delta})[+,\times]$. Therefore, \(
        \NC^1 = \mathcal{L}(\text{lex-}(\Gamma^{S_5}_{1,\delta} \cup \{\text{Maj}, \text{Sq}\})[<]).
    \) Moreover, we may restrict this to just quantifiers $\Gamma^{S_5,A}$ with $A$ a singleton set, as for any $A \subseteq S_5$, the quantifier $\Gamma^{S_5,A}_{1,\delta}$ may be easily defined using Boolean combinations of quantifiers $\Gamma^{S_5, s}_{1,\delta}$ since $S_5$ is finite. Theorem \ref{thm:logcircalgequiv} then gives us the algebraic characterization of
    \begin{align*}
        \NC^1 &=\mathcal{L}(\text{sbpc}_<(\{(\mathbb{Z}, \mathbb{Z}^+, \pm 1), (\mathbb{N}, \mathbb{S}, \{0,1\})\} \cup \{(S_5, s, S_5) \mid \{s\} \in \wp(S_5)\})).
    \end{align*}
    Because $(S_5, s, S_5) \prec (S_5, \wp(S_5), S_5)$ for all $\{s\} \in \wp(S_5)$, we lose no expressive power by replacing all elements of the form $(S_5, s, S_5)$ with $(S_5, \wp(S_5), S_5)$. We neither gain expressive power because $\mathcal{L}((S_5, \wp(S_5), S_5)) \subseteq \textsc{Reg} \subseteq \textsc{NC}^1$.  Therefore, we have our final characterization: \[
        \NC^1  = \mathcal{L}(\text{sbpc}_<(\{(\mathbb{Z}, \mathbb{Z}^+, \pm 1), (\mathbb{N}, \mathbb{S}, \{0,1\}), (S_5, \wp(S_5), S_5)\})).
    \]
\end{proof}

\section{For-Programs and Non-Lexicographic Interpretations}\label{sec:nonlex}


The proof of Lemma~\ref{lem:nesting} relies crucially on the fact that we only allow lexicographic ordering of tuples in our interpretations.  This is sufficient for the algebraic characterization in Section~\ref{sec:alg}.  However, in this section, we prove a more general version: we prove that even when interpretations are allowed to use any first-order definable ordering of tuples, we can replace multiplication quantifiers of higher arity by unary quantifiers.  To do so, we rely on the  work of Boja\'{n}czyk et al. \cite{bojanczyk2019string}, which characterizes first-order definable orders in terms of \emph{$d$-enumerators} and \emph{for-programs} which we review below.  These are combined with techniques introduced in the proof of Theorem~\ref{thm:finitebinding} to prove the following:

\begin{theorem}\label{thm:morethanlex}
    Let $\mathfrak{Q}$ be any collection of quantifiers and $\mathfrak{N}$ any collection of numerical predicates. Then, \(\mathcal{L}(\emph{\text{fo-}}(\mathfrak{Q}  \cup \Gamma^\fin)[\mathfrak{N}]) = \mathcal{L}(\emph{\text{lex-}}(\mathfrak{Q}  \cup \Gamma^\fin_1)[\mathfrak{N} \cup \{<\}]).\)
  \end{theorem}
As we noted in Section~\ref{sec:prelim}, we always assume that the collection $\mathfrak{N}$ of numerical predicates contains the order relation $<$, except in the case where it is empty.  Thus, the statement of the theorem notes that in translating our formulas to use unary quantifiers, we may need to introduce the order symbol if  $\mathfrak{N} = \emptyset$.  From the theorem, we then get the immediate corollaries.
\begin{corollary}
    \(
        \NC^1 = \mathcal{L}(\emph{\text{fo-}}(\Gamma^\fin)[+,\times]).
    \)
\end{corollary}
\begin{corollary}
    \emph{\textsc{Reg}} $= \mathcal{L}(\emph{\text{fo-}}(\Gamma^{\fin})[<])$.
\end{corollary}

Before going into the proof, we first introduce for-programs. As the definition of $d$-enumerators and for-programs in general are relatively intuitive, we keep the definition brief and a more detailed treatment may be found in \cite{bojanczyk2019string, bojanczyk2022transducers}. A \emph{first-order $d$-enumerator} is a  for-program which takes a structure $\mathfrak{A} = (A, <, \tau^\mathfrak{A})$ over vocabulary $\tau$, as input and outputs an enumeration of all the  $d$-tuples of $A$.  The program  consists of a nesting of \emph{for-loops} with a body:
\begin{equation}
  \label{eqn:for-program}
\hbox{
\begin{lstlisting}
    for $y_1$ in $p_1$
        for $y_2$ in $p_2$
            ...
                for $y_{d'}$ in $p_{d'}$
                    body
\end{lstlisting}
}
\end{equation}
Here the $i$th loop iterates the variable $y_i$ over the domain $A$ and $p_i$ determines whether the iteration is in increasing order  ($p_i = \texttt{first..last}$) or decreasing order ($p_i = \texttt{last..first}$), according to the order $<$.

The body consists of a sequence of \emph{if-statements}
\begin{equation}
  \label{eqn:for-body}
\hbox{
\begin{lstlisting}
    if $\theta_1(y_1,\dots,y_{d'})$ then
        output ${(y_{i_1^1},\dots,y_{i_{d}^1})}$
    if $\theta_2(y_1,\dots,y_{d'})$ then
        output ${(y_{i_1^2},\dots,y_{i_{d}^2})}$
    ...
    if $\theta_l(y_1,\dots,y_{d'})$ then
        output ${(y_{i_1^l},\dots,y_{i_{d}^l})}$    
\end{lstlisting}
 }
\end{equation}
Here $\theta_1,\dots,\theta_l$ are $(\FO)[<]$-formulas over the vocabulary $\tau$, specifying mutually exlusive conditions, so that for each assignment of values to the variables $y_1,\ldots,y_{d'}$ at most one of them is satisfied.  Note that which formula is satisfied may depend on the assignment  $\alpha$ of values to free variables other than $y_1,\ldots,y_{d'}$.  We only consider programs which produce, over the course of the iteration, all tuples in $A^d$, with each tuple being output exactly once.  For a string $\mathfrak{A}$ and assignment $\alpha$, we write $P(\mathfrak{A}, \alpha)$ to mean the output of the program $P$ when run on the string $\mathfrak{A}$ with assignment $\alpha$.

From Boja\'{n}czyk et al. \cite[Theorem 12]{bojanczyk2019string}, we know the following fact: 
\begin{lemma}[\cite{bojanczyk2019string}]\label{lem:bojanczyk}
    For every $(\FO)[<]$ definable linear order on $d$-tuples, defined by a formula $\phi_<$ over a vocabulary $\tau$ of unary predicates, there exists a first-order $d$-enumerator $P$ such that for every $\tau$-structure $\str{A} = (A, <, \tau^\str{A})$, $P$ on input $\str{A}$ enumerates the elements of $A^d$ in the order defined by $\phi_<$.
\end{lemma}

We now use this to prove Theorem~\ref{thm:morethanlex}.  We first use Lemma~\ref{lem:bojanczyk} to transform a formula of $(\mathfrak{Q}  \cup \Gamma^\fin)[\mathfrak{N}]$ into one in which all applications of quantifiers in $\Gamma^{\fin}$ use interpretations using a \emph{generalized lexicographic order}, and then use the techniques introduced in Section~\ref{sec:mult} to transform this to a formula using only unary quantifiers.

\begin{definition}\label{def:gen-lex}
  Given a set $A$ with a linear order $<$ on its elements and $d \in \mathbb{Z}^+$, a linear order $\prec$ on $A^d$ is called a \emph{generalized lexicographic order} if there is a sequence $\mathrm{dir} \in \{l,r\}^d$ such that $\vec{a} \prec \vec{b}$ if, and only if, for the least $i$ for which $a_i \neq b_i$, we have $a_i < b_i$ if $\mathrm{dir}_i = l$ and $a_i > b_i$ if $\mathrm{dir}_i = r$.
\end{definition}
Note that the standard lexicographic order is a generalized lexicographic order with $\mathrm{dir} = l^d$.

\begin{proof}[Proof of Theorem~\ref{thm:morethanlex}]
Let $\mathfrak{Q}$ be any collection of quantifiers and $\mathfrak{N}$ any collection of numerical predicates. Let $M = \{m_1, \dots, m_c\}$ be a finite monoid, $B \subseteq M$, $\gamma : \{0,1\}^k \rightarrow M$, and $\phi_<$ a $(\FO)[<]$-formula over a vocabulary $\tau$ which defines a linear ordering on $d$-tuples. Let $\phi_1, \dots, \phi_k$ be $\text{fo-}(\mathfrak{Q} \cup \Gamma^\fin)[\mathfrak{N}]$-formulas over $\tau$ giving an interpretation $I = (\phi_<, \phi_1,\dots,\phi_k)$ of dimension $d$. We want to show that for 
$$\Phi := \Gamma^{M,B}_{d,\gamma}\tup{x}\tup{y}(\phi_<(\tup{x}),\phi_1(\tup{y}), \dots, \phi_k(\tup{y}))[\alpha],$$
there is an equivalent formula in $\text{lex-}(\mathfrak{Q} \cup \Gamma^\fin_1)[\mathfrak{N} \cup \{<\}]$.
We first prove by induction on the nesting of multiplication quantifiers that $\Phi$ is equivalent to a formula in which all applications of quantifiers in $\Gamma^{\fin}$ are to interpretations in which the order defined is a generalized lexicographic order.  Note that formulas (e.g., $\phi_{<}$) defining the order in interpretations are in $(\FO)[<]$ and, thus, by assumption do not have nested multiplication quantifiers.

For the base case, say we have a formula which uses no multiplication quantifiers; then, we are done as this is a formula of $\text{lex-}(\mathfrak{Q} \cup \Gamma^\fin_1)[\mathfrak{N} \cup \{<\}]$.
For the inductive step,  assume that $\phi_1, \dots, \phi_k$ are formulas of $\text{fo-}(\mathfrak{Q} \cup \Gamma^\fin_1)[\mathfrak{N} \cup \{<\}]$ in which all orders are generalized lexicographical orders.

By Lemma~\ref{lem:bojanczyk}, we have a first-order $d$-enumerator $P$ which enumerates the elements of $A^d$ in the order defined by $\phi_<$ on any $\tau$-structure $\mathfrak{A}$.  $P$ has the form specified in~(\ref{eqn:for-program}) and~(\ref{eqn:for-body}). Recall that in this for-program we have $l$ if-statements. Now, let $\Phi'$ be the formula \begin{align*}
    \Phi' := \Gamma^{M,B}_{d',\delta} \tup{x} y_1 \dots y_{d'} (\xi_<(\tup{x}), 
        &\xi_1(y_1,\dots,y_{d'}), \dots, \xi_c(y_1, \dots, y_{d'}), \\
        &\xi_{c+1}(y_1, \dots, y_{d'}), \dots, \xi_{2c}(y_1, \dots, y_{d'}), \\
        &\dots, \\
        &\xi_{(l-1)c+1}(y_1, \dots, y_{d'}), \dots, \xi_{lc}(y_1, \dots, y_{d'})
    )
\end{align*}
where $\delta : \{0,1\}^{lc} \rightarrow M$ is such that for a word $w = w_1\dots w_{lc} \in \{0,1\}^{lc}$, $\delta(w) = m_i$ if there exists a $j$ where $0 \leq j < l$ such that the substring $w_{jc+1} \dots w_{(j+1)c}$ is equal to the one-hot encoding of $i$ and all other characters in $w$ are $0$. In other words, partitioning $w$ into $l$ consecutive $c$-length substrings, we map $w$ to $m_i$ if one of these substrings is the one-hot encoding of $i$ and all others are simply strings of $0$s. All other elements of the domain are mapped to the identity of $M$. Moreover, for $0 \leq j < l$ and $1 \leq j' \leq c$, $\xi_{jc+j'}(y_1, \dots, y_{d'})$ is the formula $\theta_j(y_1, \dots, y_{d'}) \wedge \psi_{m_{j'}}(y_{i^j_1}, \dots, y_{i^j_d})$ where ${i^j_1}, \dots, {i^j_d}$ are defined as they are in the for-program $P$, and $\psi_{m_{j'}}$ is as defined in Lemma~\ref{lem:finitetuple}; therefore, $\xi_{jc+j'}$ is satisfied if  $\theta_j(y_1, \dots, y_{d'})$ is satisfied and the element of $M$ indexed by $(y_{i^j_1}, \dots, y_{i^j_d})$ equals $m_{j'}$. Lastly, $\xi_<$ is the generalized lexicographic order given by the vector $\mathrm{dir} \in \{l,r\}^{d'}$ with $\mathrm{dir}_i = l$ if  $p_{i} = \texttt{first..last}$ and $\mathrm{dir}_i = r$ if $p_{d'} = \texttt{last..first}$.
We now prove that $\mathfrak{A} \models \Gamma^{M,B}_{d,\gamma}\tup{x}\tup{y}(\phi_<(\tup{x}),\phi'_1(\tup{y}), \dots, \phi'_k(\tup{y}))[\alpha]$ if, and only if, $\mathfrak{A} \models \Phi'[\alpha]$.

Let $1_M$ denote the identity of $M$ and $f_{\mathfrak{A},\alpha} : A^d \cup \{\epsilon\} \rightarrow M$ be the function defined by 
$$f_{\mathfrak{A},\alpha}(a_1, \dots, a_d) = \gamma({\phi'_1}^{\mathfrak{A},\alpha}[a_1,\dots,a_d],\dots,{\phi'_k}^{\mathfrak{A},\alpha}[a_1,\dots,a_d])$$ for $(a_1, \dots, a_d) \in A^d$ and $f_{\mathfrak{A},\alpha}(\epsilon) = 1_M$.
Let $g_{\mathfrak{A},\alpha} : A^{d'} \rightarrow A^d \cup \{\epsilon\}$ be the function defined by the body of the for-program $P$ where if no output is produced, $g_{\mathfrak{A},\alpha}(a_1, \dots, a_{d'}) = 0^d$; thus, $$
    g_{\mathfrak{A},\alpha}(a_1, \dots, a_{d'}) = \begin{cases}
        (a_{i^j_1}, \dots, a_{i^j_d}) & \text{if } \exists j \text{ such that } \mathfrak{A} \models \theta_j(a_1, \dots, a_{d'})\\
        \epsilon & \text{o.w.}
    \end{cases}
$$
Let $f^*_{\mathfrak{A},\alpha} : (A^d)^* \rightarrow M$ and $g^*_{\mathfrak{A},\alpha} : (A^{d'})^* \rightarrow (A^d)^*$ be the monoid homomorphisms induced by the functions $f_{\mathfrak{A},\alpha}$ and $g_{\mathfrak{A},\alpha}$, respectively, in the natural way.

Let $\mathfrak{A}$ be an arbitrary $\tau$-structure and $\alpha$ an arbitrary variable assignment with assignments for at least all free variables in each $\phi'_i$ with the exception of $\tup{y}_i$. By definition of multiplication quantifiers, we get that
$$\mathfrak{A} \models \Gamma^{M,B}_{d,\gamma}\tup{x}\tup{y}(\phi_<(\tup{x}),\phi'_1(\tup{y}), \dots, \phi'_k(\tup{y}))[\alpha] \text{ iff } \gamma(w_{I(\mathfrak{A},\alpha)}) \in B.$$


By construction of $f_{\mathfrak{A},\alpha}$ and because $P$ provides the order by which $\Gamma^{M,B}_{d,\gamma}$ is evaluated, it is easy to see that $\gamma(w_{I(\mathfrak{A},\alpha)}) \in B \text{ iff } f^*_{\mathfrak{A},\alpha}(P(\mathfrak{A}, \alpha)) \in B$.

Let $t_1, \dots, t_{|A|^{d'}}$ denote the ordering of $d'$-tuples of $A$ as defined by $\xi_<$. Because the output of $P$ is an enumeration of $A^d$ in the ordering defined by $\phi_<$, and $g_{\mathfrak{A},\alpha}$ outputs the empty string $\epsilon$ during an iteration of its for-loops only when $P$ outputs nothing at that iteration, it follows that $$
    f^*_{\mathfrak{A},\alpha}(g^*_{\mathfrak{A},\alpha}(t_1\dots t_{A^{d'}})) = f^*_{\mathfrak{A},\alpha}(P(\mathfrak{A},\alpha))$$
and, thus, $f^*_{\mathfrak{A},\alpha}(g^*_{\mathfrak{A},\alpha}(t_1\dots t_{A^{d'}})) \in B$ iff $f^*_{\mathfrak{A},\alpha}(P(\mathfrak{A},\alpha)) \in B$.

Let $I' = (\xi_<, \xi_1, \dots, \xi_{lc})$. We now want to show that 
$$\delta(w_{I'(\mathfrak{A},\alpha)}) \in B \text{ iff } f^*_{\mathfrak{A},\alpha}(g^*_{\mathfrak{A},\alpha}(t_1\dots t_{A^{d'}})) \in B.$$ 
To do so, we prove that $\delta((w_{I'(\mathfrak{A},\alpha)})_a) = f^*_{\mathfrak{A},\alpha}(g^*_{\mathfrak{A},\alpha}(t_a))$ for every $0 \leq a \leq |A|^{d'}$. Let $a$ be arbitrary. Then, $u = (w_{I(\mathfrak{A},\alpha)})_a$ is either $0^{lc}$ or a one-hot encoding of some $b \in [lc]$. If $u = 0^{lc}$, then $\delta(u) = 1_M$ and no $\theta_j$ is satisfied. Thus, $g^{*}_{\mathfrak{A},\alpha}(t_a) = \epsilon$ and $f^*_{\mathfrak{A},\alpha}(\epsilon) = 1_M$ by construction. If $u$ is a one-hot encoding of some $b \in [lc]$, then some $\theta_j$ is satisfied and, 
$$\delta(u) = m = \gamma({\phi'_1}^{\mathfrak{A},\alpha}[a_{i^j_1},\dots,a_{i^j_d}],\dots,{\phi'_k}^{\mathfrak{A},\alpha}[a_{i^j_1},\dots,a_{i^j_d}]) = f_{\mathfrak{A},\alpha}(a_{i^j_1}, \dots, a_{i^j_d})$$
 by definition of $\delta$ and the construction of each $\xi_i$. By the construction of $g_{\mathfrak{A},\alpha}$ and $f_{\mathfrak{A},\alpha}$, we also get immediately that $g^*_{\mathfrak{A},\alpha}(t_a) = (a_{i^j_1}, \dots, a_{i^j_d})$ and, thus, $f^*_{\mathfrak{A},\alpha}(g^*_{\mathfrak{A},\alpha}(t_a)) = m$. Thus, $\delta(w_{I'(\mathfrak{A},\alpha)}) \in B$ if, and only if, $f^*_{\mathfrak{A},\alpha}(g^*_{\mathfrak{A},\alpha}(t_1 \dots t_{|A|^{d'}})) \in B$.

Then, by definition, we have that $\mathfrak{A} \models \Phi'[\alpha]$ if, and only if, $\delta(w_{I'(\mathfrak{A},\alpha)}) \in B$.

Therefore, all together, we have that \begin{align*}
    &\mathfrak{A} \models \Gamma^{M,B}_{d,\gamma}\tup{x}\tup{y}(\phi_<(\tup{x}), {\phi'_1}(\tup{y}), \dots, \phi'_k(\tup{y}))[\alpha]\\
    \text{ iff } &\gamma(w_{I(\mathfrak{A},\alpha)}) \in B\\
    \text{ iff } &f^*_{\mathfrak{A},\alpha}(P(\mathfrak{A},\alpha)) \in B\\
    \text{ iff } &f^*_{\mathfrak{A},\alpha}(g^*_{\mathfrak{A},\alpha}(t_1\dots t_{A^{d'}})) \in B\\
    \text{ iff } &\delta(w_{I'(\mathfrak{A},\alpha)}) \in B\\
    \text{ iff } &\mathfrak{A} \models \Phi''[\alpha].
\end{align*}

This concludes the inductive step. Therefore, for any formula $\Phi$, we have an equivalent formula $\Phi'$ which uses only uses multiplication quantifiers using a generalized lexicographic order. 

To see that we can construct a formula equivalent to $\Phi'$ using only unary quantifiers, we use the technique used in the proof of Lemma~\ref{lem:nesting}  with a minor modification to account for the generalized lexicographic order. In the proof of Lemma~\ref{lem:nesting}, a quantifier $\Gamma^{M,B}_{d,\gamma}$ of dimension $d$ is replaced by a sequence of $d$ nested unary quantifiers over the same monoid $M$.  
To account for the order $\xi_<$, which defines a generalized lexicographic order with $\mathrm{dir} \in \{l,r,\}^d$, we define the $i$th quantifier in the sequence as before if $\mathrm{dir}_i = l$ and replace the monoid $(M,\cdot)$ with $(M, \cdot^R)$ where $m \cdot^R m' := m' \cdot m$ for all $m,m' \in M$ if $\mathrm{dir}_i = r$.  The proof follows exactly along the lines of the proof of  Lemma~\ref{lem:nesting}.

All together, we have an equivalent formula to $\Phi'$, and thus $\Phi$, which only uses unary quantifiers, completing the proof.  
\end{proof}

\section{Conclusion}\label{sec:conc}

In this work, we constructed a class of typed monoids exactly recognizing $\NC^1$. To do so, we proved results regarding the expressive power of logics with quantifiers defined over finite monoids.  Specifically, we established that the expressive power is not changed by restricting the dimension of the interpretations on which the quantifiers act, regardless of which numerical predicates are available.

Therefore, we were able to provide a logic characterizing $\NC^1$ which only uses unary quantifiers and use this logic to construct an algebraic characterization of $\NC^1$.  This result marks the second circuit complexity class to be characterized in a such a way, with the first being $\TC^0$ \cite{krebs2007characterizing}, and provides a decomposition theorem in the style of Krohn-Rhodes for $\NC^1$.

An interesting future direction would be to  construct similar algebraic characterization of other complexity classes beyond $\NC^1$.  It seems this would require the development of new algebraic tools.  In particular, the block product is a key tool used to characterize first-order quantification and more generally quantification over interpretations of dimension one.  To extend the work to other complexity classes, it would be worthwhile investigating other product constructions that might similarly relate to first-order quantifiiers of higher arity as well as higher-order quantifiers.

Moreover, we extended the work of Barrington et al. \cite{barrington1990uniformity} and Lautemann et al. \cite{lautemann2001descriptive} by studying the expressive power of multiplication quantifiers when applied to interpretations using other than the standard lexicographic order.  We showed that over strings the expressive power of logics equipped with finite multiplication quantifiers is not changed by loosening the restriction to permitting any first-order definable linear order. A natural next step would be to investigate the expressive power when other linear orders are permitted. For example, by extending the Domination Lemma of Boja\'{n}czyk et al. \cite{bojanczyk2019string} to monadic second-order logic (MSO[$<$]), one would be able to show that permitting the application of multiplication quantifiers to interpretations with an MSO[$<$]-definable order does not change the expressive power. This would result in the interesting consequence that $$\textsc{Reg} = \mathcal{L}(\text{lex-}(\Gamma^{\fin}_1)[<]) = \mathcal{L}((\Gamma^{\fin})[<]).$$

\printbibliography

\end{document}